\documentclass{article}

\PassOptionsToPackage{numbers,compress}{natbib}
\usepackage[preprint]{neurips_2026}

\usepackage[utf8]{inputenc}
\usepackage[T1]{fontenc}
\usepackage{hyperref}
\usepackage{url}
\usepackage{booktabs}
\usepackage{amsfonts}
\usepackage{nicefrac}
\usepackage{microtype}
\usepackage{xcolor}
\usepackage{amsmath}
\usepackage{amssymb}
\usepackage{graphicx}
\usepackage{afterpage}
\usepackage{amsthm}
\usepackage{multirow}
\usepackage{wrapfig}
\usepackage{subcaption}
\usepackage{placeins}
\usepackage{float}
\usepackage[most]{tcolorbox}

\newtheorem{theorem}{Theorem}

\newtheorem{remark}{Remark}

\definecolor{theoremblue}{RGB}{33,92,168}
\definecolor{remarkgray}{RGB}{120,120,120}

\newtcolorbox{maintheorembox}{
  enhanced,
  breakable,
  colback=theoremblue!4,
  colframe=theoremblue!70!black,
  boxrule=0.85pt,
  arc=3pt,
  left=10pt,
  right=10pt,
  top=8pt,
  bottom=8pt,
  before skip=10pt,
  after skip=10pt,
  borderline west={1.6pt}{0pt}{theoremblue!80!black}
}

\newtcolorbox{mainremarkbox}{
  enhanced,
  breakable,
  colback=black!3,
  colframe=remarkgray!75!black,
  boxrule=0.7pt,
  arc=3pt,
  left=10pt,
  right=10pt,
  top=7pt,
  bottom=7pt,
  before skip=8pt,
  after skip=10pt,
  borderline west={1.2pt}{0pt}{remarkgray!85!black}
}

\newtcblisting{promptbox}[1]{
  enhanced,
  breakable,
  colback=black!2,
  colframe=remarkgray!75!black,
  boxrule=0.65pt,
  arc=3pt,
  left=8pt,
  right=8pt,
  top=7pt,
  bottom=7pt,
  before skip=8pt,
  after skip=10pt,
  title={#1},
  fonttitle=\bfseries,
  listing only,
  listing options={
    basicstyle=\ttfamily\footnotesize,
    breaklines=true,
    columns=fullflexible,
    keepspaces=true,
    showstringspaces=false
  }
}

\title{ARMOR: Adaptive Retriever Optimization for Low-Resource Telecom Question Answering}

\author{%
Heshan Fernando\textsuperscript{1} \quad
Quan Xiao\textsuperscript{2} \quad
Yan Xin\textsuperscript{3} \quad
Tianyi Chen\textsuperscript{2, 1}\\
\textsuperscript{1}Rensselaer Polytechnic Institute, Troy, NY\\
\textsuperscript{2}Cornell University, New York, NY\\
\textsuperscript{3}Samsung Research America, Berkeley Heights, NJ\\
\texttt{fernah@rpi.edu},\\
\texttt{\{qx232,tianyi.chen\}@cornell.edu},\\
\texttt{yan.xin@samsung.com}
}

\begin{document}

\maketitle

\begin{abstract}
Telecom question answering (QA) is a challenging setting for retrieval-augmented generation (RAG): evidence is fragmented across standards, papers, encyclopedic resources, and web documents, and answers often hinge on technical tables, equations, and specialized protocol language. In low-resource subdomains, generator fine-tuning can over-specialize and degrade general capability, making query-side retriever adaptation an attractive alternative. To this end, we ask whether a fixed-generator, query-adapted RAG system can outperform generator-side adaptation, and which retriever objectives best support that setting. We motivate retrieval, rather than generator fine-tuning, as the adaptation target through a capacity comparison: under bounded-parameter and soft-retrieval assumptions, query-encoder tuning can have a smaller estimation term than supervised fine-tuning when its effective dimension is smaller. We identify two particularly relevant objectives---the latent-document RAG likelihood, which optimizes generation utility, and the InfoNCE contrastive objective, which improves semantic retrieval geometry---and leverage them jointly through a retriever optimization method targeting downstream QA performance in the telecom domain. Specifically, we introduce ARMOR, Adaptive Regularized Mixture Optimization for Retrievers, which learns separate temperatures for the RAG retrieval distribution and InfoNCE softmax and regularizes the adapted query encoder toward the frozen base query encoder. Across telecom-specific retrieval and generative QA benchmarks, we show that ARMOR improves evidence retrieval and answer generation in several in-domain settings. Code is available at \url{https://github.com/heshandevaka/ARMOR.git}.
\end{abstract}

\section{Introduction}\label{sec:intro}
 
\begin{figure}[t]
    \centering
    \includegraphics[width=0.9\textwidth]{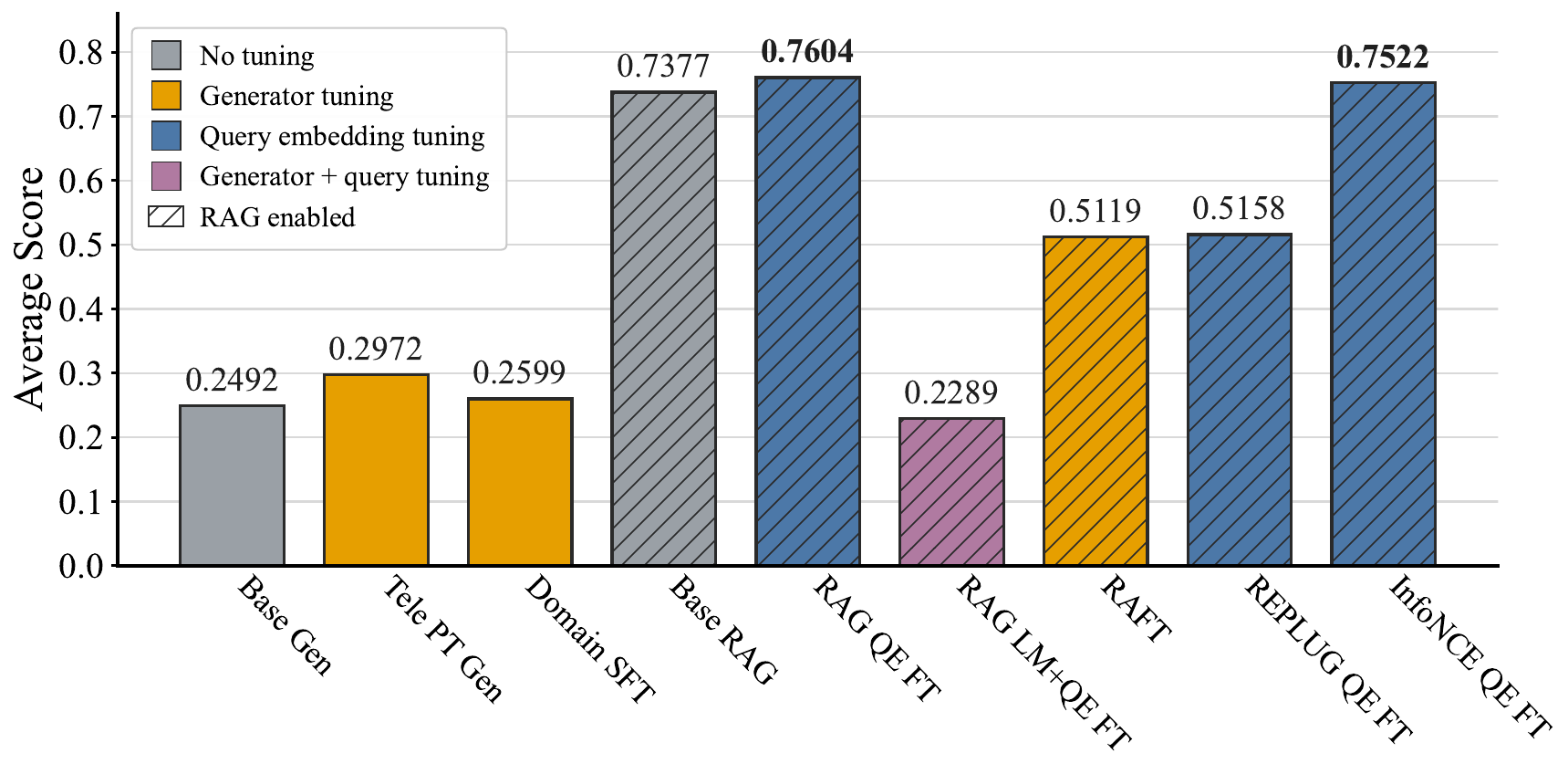}
    \caption{\textbf{Comparison of optimization targets and the performance of retriever-training objectives for ISAC domain QA} with Llama-3-8B-Instruct \citep{grattafiori2024llama} as the generator model and e5-large-v2 \citep{wang2022text} as the dense retriever backbone. Retriever-side query-encoder optimization produces substantially larger in-domain gains than either the base generator or generator-side adaptation, motivating our study of which component to optimize and which retriever objective best improves RAG performance. The compared methods are \textit{Base Gen}, the closed-book generator without retrieval; \textit{Tele PT Gen}, a telecom-domain pretrained generator model~\cite{maatouk2024tele}; \textit{Domain SFT}, a domain-specific supervised fine-tuned generator model; \textit{Base RAG}, retrieval-augmented generation with the frozen base retriever; \textit{RAG QE FT}, query-encoder fine-tuning under the original RAG marginal likelihood objective~\cite{lewis2020rag}; \textit{REPLUG QE FT}, query-encoder fine-tuning using the REPLUG objective~\cite{shi2024replug}; \textit{RAFT}, query-encoder fine-tuning using the RAFT objective~\cite{zhang2024raft}; \textit{InfoNCE QE FT}, query-encoder fine-tuning under the supervised contrastive InfoNCE objective~\cite{izacard2021contriever}.}
    \label{fig:intro_objective_tradeoffs}
    \vspace{-0.2cm}
\end{figure}%

Telecom question answering (QA) is challenging for off-the-shelf LLMs because relevant evidence is fragmented across standards, research papers, encyclopedic resources, and web documents, while correct answers often depend on equations, tables, protocol logic, and cross-references. Recent resources such as Tele-Data and Tele-Eval make this setting tractable by pairing telecom corpora with QA supervision and provenance, naturally favoring grounded retrieval-augmented generation (RAG) over purely parametric knowledge storage \cite{maatouk2024tele}. In low-resource subdomains, however, full generator adaptation is brittle: labeled data are scarce, sequential tuning risks catastrophic forgetting \cite{kirkpatrick2017ewc}, and prior telecom studies suggest that LoRA fine-tuning may provide limited improvements when the base model contains only partial domain knowledge \cite{hu2021lora,maatouk2024tele}. RAG offers a more modular alternative: when the generator is already partially capable, improving which evidence it sees can be more effective than changing the generator itself \cite{lewis2020rag}.

We therefore study a simple but consequential question: when domain supervision is scarce and the document index is fixed, can query-side retriever adaptation beat generator tuning, and which retrieval objective should drive it? Figure~\ref{fig:intro_objective_tradeoffs} previews the empirical answer on the Integrated Sensing and Communication (ISAC) domain: query-encoder optimization produces larger in-domain gains than either closed-book generation or generator-side adaptation. This motivates a retriever-centric strategy for low-resource telecom QA: keep the generator fixed, adapt only the query interface, and preserve the frozen document index used at inference time.

However, choosing the retriever as the adaptation target does not by itself determine how to train it. The latent-document RAG likelihood rewards passages that improve answer generation \cite{lewis2020rag}, while InfoNCE improves semantic separation and retrieval discrimination \cite{karpukhin2020dpr,izacard2021contriever,oord2018representation}. These signals are complementary, but they also raise two design challenges. First, a single objective can overemphasize either downstream utility or embedding geometry, and a static mixture assumes the right balance is fixed throughout training. Second, because only the query encoder is updated while the document encoder and index remain frozen, retriever fine-tuning can move query embeddings away from the base document space. In low-data settings, this drift can improve the training objective without reliably improving evidence coverage at test time.

We address these issues with \textbf{ARMOR} (Adaptive Regularized Mixture Optimization for Retrievers). ARMOR learns separate temperatures for the RAG retrieval distribution and the InfoNCE softmax, allowing each objective's sharpness and influence to change during training, and regularizes the adapted query encoder toward the frozen base query encoder to preserve compatibility with the fixed document space. In the main source-document-split Tele-Eval setting, ARMOR improves the ISAC answer score from 0.6893 to 0.7119 while also achieving the best Recall@3 and Recall@5; across ISAC, JCC, and SAGIN, it improves evidence coverage at higher recall ranks while avoiding the answer-quality degradation seen in other fine-tuned retriever baselines.

\noindent\textbf{Contributions.}
The main contributions are:
\begin{itemize}
    \setlength{\itemsep}{0.25em}
    \setlength{\parsep}{0pt}
    \setlength{\parskip}{0pt}
    \item[C1)] \textbf{Retriever-centric adaptation for low-resource telecom QA.}
    We provide a capacity-based motivation and empirical evidence that, when domain supervision is scarce, query-encoder tuning can offer a favorable estimation--performance tradeoff relative to generator fine-tuning while preserving the modularity of a fixed-generator, fixed-index RAG system.

    \item[C2)] \textbf{Adaptive objective balancing for retriever optimization.}
     We show that RAG likelihood and InfoNCE capture complementary retriever signals, and introduce objective-specific learnable temperatures so their influence can change during training without manually selecting a fixed mixture weight.

    \item[C3)] \textbf{Base-compatible query regularization for telecom QA.}
    By distilling adapted query embeddings toward the frozen base query encoder, ARMOR limits drift from the fixed document embedding space. Across telecom retrieval and QA benchmarks, ARMOR gives the strongest overall tradeoff across in-domain retrieval quality, downstream answer quality, and robustness across generator scales.
\end{itemize}

\section{What to Optimize in Low-Resource Telecom QA?}

\subsection{Query-Only RAG Setup and Retriever-Centric Adaptation}

Let $\mathcal{D} = \{d_j\}_{j=1}^{N}$ be a datastore of chunked telecom passages derived from standards, research papers, and curated domain web sources. Given a telecom question $x$, the retriever returns a top-$k$ set $\mathcal{N}_k(x)=\{d_i\}_{i=1}^{k} \subset \mathcal{D}$, and a generator produces an answer $y$ conditioned on $x$ and the retrieved evidence. We adopt a dense retrieval model with score
\[
s_{\eta}(x,d) = f_{\eta}(x)^\top g(d),
\]
where $f_{\eta}$ is a trainable query encoder and $g$ is a document encoder. In the most operationally attractive variant, $g$ is fixed and the corpus is pre-embedded, so adaptation occurs only on the query side. We denote the retriever embedding dimension by $d_E$, so that $f_{\eta}(x),g(d)\in\mathbb{R}^{d_E}$.

This query-only setup is well matched to the low-resource telecom regime. In narrow telecom slices, the generator is often partially capable but fails when shown incorrect or incomplete evidence. Updating the retriever is therefore a modular intervention that focuses limited supervision on the component controlling evidence exposure, while avoiding repeated re-indexing and preserving the base generator's general capabilities.

This perspective is supported by the empirical pattern previewed in Figure~\ref{fig:intro_objective_tradeoffs}: retriever-side optimization produces stronger in-domain gains than either the base model or generator-side fine-tuning. More importantly, retriever-centric learning aligns with operational constraints. It allows the system to update evidence access without repeatedly changing the generator, reuses a frozen document index when only the query encoder is tuned, and keeps grounding behavior interpretable through retrieved passages.

We next make this architectural choice more explicit by comparing the low-data generalization behavior of generator fine-tuning and retriever fine-tuning.

\subsection{Low-Data Generalization Motivation}

The next result is used as a capacity comparison. It isolates one regime in which query-side tuning can have a smaller estimation term than generator fine-tuning, namely when the retriever's effective dimension is smaller under bounded-parameter and soft-retrieval assumptions.

We compare two families of adaptation, \textit{SFT} and \textit{RAG}, on the same population distribution $\mathcal{P}$ over examples $z=(x,y)$, where $x$ is the question, and $y$ is the target answer. Given $N$ i.i.d. training pairs $\{z_i\}_{i=1}^N$ drawn from the population distribution $\mathcal{P}$, let $\theta\in\mathcal{E}_S\subset\mathbb{R}^{d_S}$ be the trainable generator parameter for SFT, $\eta\in\mathcal{E}_R\subset\mathbb{R}^{d_R}$ be the trainable RAG parameters, and $\ell_S(z ; \theta)$ and $\ell_R(z ; \eta)$ be the per-sample SFT and RAG loss, respectively.

The population SFT loss is defined as the expected test loss under the true population distribution $\mathcal{P}$, while the empirical SFT loss is defined as the empirical training loss on finite sample $\{z_i\}_{i=1}^N$, i.e. 
\begin{align*}
\textbf{population loss: } L_S(\theta)=\mathbb{E}_{z \sim \mathcal{P}}\left[\ell_S(z ; \theta)\right], ~~ \textbf{empirical loss: } \widehat{L}_{S, N}(\theta)=\frac{1}{N} \sum_{i=1}^N \ell_S(z_i ; \theta).
\end{align*}
Similarly, we can define the population and empirical RAG losses as follows. 
\begin{align*}
&\textbf{population loss: } L_R(\eta)=\mathbb{E}_{z \sim \mathcal{P}}\left[\ell_R(z ; \eta)\right], ~~ \textbf{empirical loss: } \widehat{L}_{R, N}(\eta)=\frac{1}{N} \sum_{i=1}^N \ell_R(z_i ; \eta). 
\end{align*}

A standard goal of generalization theory is to control the gap between the population loss and the empirical loss. In our comparison, a smaller upper bound should be interpreted as better estimation control under the stated assumptions rather than as a direct prediction of downstream performance.

We denote the RAG answer distribution as $q_\eta(\cdot ~|~ x)=\sum_{d \in \mathcal{N}_k(x)} p_\eta(d \mid x) p_{\theta_0}(\cdot \mid x, d)$.
Therefore, letting $\ell_{\operatorname{NLL}}(q, x,y)=-\log q(y~|~x)$ be the negative log-likelihood loss, the SFT and RAG loss can be rewritten as 
\begin{align*}
\ell_S(z ; \theta)=\ell_{\mathrm{NLL}}(p_\theta, x, y), \text{ and } \ell_R(z ; \eta)=\ell_{\mathrm{NLL}}(q_\eta, x, y).
\end{align*}
where the first optimizes the generator parameter $\theta$ to predict the answer $y$ given the question $x$, and the latter optimizes the retrieval parameter $\eta$ to select the relevant document $d$, so that the model can generate the correct answer $y$ conditioned on the resulting question–document pair $(x,d)$. 

\begin{maintheorembox}
\vspace{-0.1cm}
\begin{theorem}[Capacity comparison for SFT and RAG]\label{thm:gen-final}
Assume that the loss function $\ell_{\mathrm{NLL}}(q,x,y)$ is bounded in $[0, c]$ and is $\rho$-Lipschitz with respect to $q$ in the feasible domain, and that the weights for SFT and RAG are bounded by $R$ for all $t$.
Also assume that the tokenized inputs for questions and answers $(x_n,y_n)$ have full rank and are bounded by $R_X$ after scaling by the corresponding dimensions $d_R$ and $d_S$, for all $n\in [N]$. Additionally, assume $\|g(d)\|\leq B_g$ and $q_\eta$ is $\rho^{\text{soft}}$-Lipschitz over $s_\eta (x,d)$ on the feasible domain. Then with probability at least $1-\delta$, we have for any $\theta \in \mathcal{E}_S$ and $\eta \in \mathcal{E}_R$,
\begin{align}
&L_S(\theta) \leq  \widehat L_{S,N}(\theta)+\tilde{\cal O}(\rho C(N,R,R_X,d_S,d_{m}))+3 c \sqrt{\frac{\log (2 / \delta)}{2 N}}\\
&L_R(\eta)\leq \widehat L_{R,N}(\eta)+\tilde{\cal O}(\rho^{\text{soft}}B_g\sqrt{d_E}\rho C(N,R,R_X,d_R,d_{m^\prime}))+3 c \sqrt{\frac{\log (2 / \delta)}{2 N}}
\end{align}
where $C(N,R,R_X,d,d_{m}):=\sqrt{\frac{P(d,d_{m})}{N^3}}\left(1+\log \left(R_O R_V(\sqrt{d} R_X) \sqrt{\frac{N}{P(d,d_{m})}}\right)\right)$ 
and
\[
P(d,d_{m})= (\sqrt{d} R_X)^2\left(\left(\sqrt{d_m} R_V\right)^{\frac{2}{3}}+\left(\sqrt{d_m} R_K R_Q R_V\right)^{\frac{2}{3}}\right)^3 \log \left(N d\right)
\]
with $d_m$ and $d_{m^\prime}$ denoting the corresponding model architecture widths for SFT and RAG. 
\end{theorem}
\end{maintheorembox}

Theorem \ref{thm:gen-final} builds on generalization theory for Transformer-based models \citep{mwigogeneralization,trauger2024sequence} and retriever tuning \citep{basu2024statistical}.
The proof of Theorem \ref{thm:gen-final} is provided in Appendix \ref{sec:theory_proof}. 

\begin{mainremarkbox}
\begin{remark}
When the empirical losses are comparable and $N$ is small, the complexity term can dominate the comparison between the two upper bounds. Omitting the logarithmic factors, the SFT complexity term scales as $\tilde{\mathcal O}(
\rho \frac{\sqrt{d_Sd_m}}{N^{3/2}}
)$, 
whereas the RAG complexity term scales as $\tilde{\mathcal O}(
\rho\rho^{\rm soft}B_g
\frac{\sqrt{d_Ed_Rd_{m'}}}{N^{3/2}}
)$. 
Therefore, when \(\rho^{\rm soft}B_g=O(1)\), the RAG upper bound can be tighter whenever $d_Ed_Rd_{m^\prime}\lesssim d_Sd_{m}$. 
\end{remark}
\end{mainremarkbox}

In our experiments, query-encoder adaptation uses about four times fewer trainable parameters than generator adaptation, which is consistent with the dimensional condition above. Theorem \ref{thm:gen-final} should therefore be read as a motivating capacity argument rather than a standalone explanation of empirical superiority. It helps explain why retriever tuning is a plausible low-data adaptation target, while the results in Figure \ref{fig:intro_objective_tradeoffs} and Section~\ref{sec:results} determine whether this advantage appears in the evaluated telecom QA setting. The remaining design question is which retriever objectives should drive this adaptation.

\subsection{Retriever Objectives: RAG Likelihood, InfoNCE, and Static Mixing}

Following the empirical pattern in Figure \ref{fig:intro_objective_tradeoffs}, we focus on two complementary retriever objectives: RAG likelihood and InfoNCE.

\paragraph{RAG likelihood.}
With a frozen generator $p_{\theta_0}(y\mid x,d)$, the RAG objective optimizes retrieval indirectly through answer likelihood:
\[
\mathcal{L}_{\mathrm{RAG}}(x,y)
=
-\log \sum_{d \in \mathcal{N}_k(x)} p_{\eta}(d\mid x)\, p_{\theta_0}(y\mid x,d),
\]
where $\mathcal{N}_k(x)$ is the top-$k$ document set ranked by the retrieval score $s_\eta(x,d)$ and
\[
p_\eta(d \mid x)=\frac{\exp \left(s_\eta(x, d)\right)}{\sum_{d^{\prime} \in \mathcal{N}_k(x)} \exp \left(s_\eta\left(x, d^{\prime}\right)\right)}.
\]
This objective rewards passages that are useful for generation, not merely passages that are topically related.

\paragraph{InfoNCE.}
Given a positive passage $d^{+}$ and a set of negatives $\mathcal{N}^{-}$, the contrastive objective is
\[
\mathcal{L}_{\mathrm{InfoNCE}}(x,d^{+},\mathcal{N}^{-})
=
-\log
\frac{\exp(s_{\eta}(x,d^{+}))}
{\exp(s_{\eta}(x,d^{+})) + \sum_{d^{-} \in \mathcal{N}^{-}} \exp(s_{\eta}(x,d^{-}))}.
\]
In telecom QA, positive passages can often be obtained from source-linked datasets such as Tele-Eval, while negatives may come from in-batch sampling, BM25, or hard-negative mining. This objective primarily improves semantic separability and retrieval discrimination.

These objectives matter for different reasons. RAG likelihood is utility-driven: it improves the retriever to the extent that retrieval helps the generator answer correctly. InfoNCE is geometry-driven: it improves the embedding space so that relevant evidence is easier to recover. In low-resource telecom QA, both signals are needed. RAG alone may under-shape the representation space, while InfoNCE alone may reward topical similarity without fully aligning retrieval with downstream answer utility. Empirically, both objectives perform well in Figure~\ref{fig:intro_objective_tradeoffs}, which motivates combining them rather than treating either objective as sufficient on its own.

A standard mixture can be written as
\[
\mathcal{L}_{\mathrm{static}} = \lambda\, \mathcal{L}_{\mathrm{RAG}} + (1-\lambda)\, \mathcal{L}_{\mathrm{InfoNCE}},
\]
where $\lambda \in [0,1]$ is fixed in advance. The static mixture has two limitations that ARMOR targets directly: it fixes the objective balance for all training stages, and it does not constrain the adapted query encoder to remain compatible with the frozen document index. 

\section{ARMOR: Adaptive and Regularized Retriever Optimization}
\label{sec:regularization}

\subsection{Temperature-Modulated Retriever Objectives}

The fixed mixture in $\mathcal{L}_{\mathrm{static}}$ assumes that the relative influence of RAG and InfoNCE should remain constant throughout training. This is a strong assumption: the two objectives capture different aspects of retriever learning, and their relative usefulness need not be the same early and late in optimization. Kendall et al.~\citep{kendall2018multitask} make a related point in multi-task learning, showing that fixed loss weighting can bias training toward one objective and that adaptive weighting can lead to more balanced optimization. 

Our setting differs in that we optimize a single retriever with two complementary signals, and we place adaptivity inside the softmax distributions that define the objectives rather than on external loss coefficients. 
We therefore introduce learnable positive temperatures into both retriever objectives. In implementation, we optimize unconstrained scalars $\alpha_r$ and $\alpha_c$ and set
\[
\tau_r=\tau_{\min}+\operatorname{softplus}(\alpha_r),\qquad
\tau_c=\tau_{\min}+\operatorname{softplus}(\alpha_c),
\]
where $\tau_{\min}>0$ is a small floor used only to avoid degenerate zero-temperature softmaxes. For RAG, the retriever distribution over the top-$k$ documents becomes
\[
p_{\eta}^{(\tau_r)}(d_i \mid x)=
\frac{\exp(s_{\eta}(x,d_i)/\tau_r)}{\sum_{j=1}^{k}\exp(s_{\eta}(x,d_j)/\tau_r)},
\]
which yields the temperature-modulated loss
\[
\mathcal{L}_{\mathrm{RAG}}(x,y;\tau_r)=
-\log \sum_{i=1}^{k} p_{\eta}^{(\tau_r)}(d_i \mid x)\, p_{\theta_0}(y \mid x,d_i).
\]
For InfoNCE, we similarly define
\[
\mathcal{L}_{\mathrm{InfoNCE}}(x,d^{+},\mathcal{N}^{-};\tau_c)=
-\log
\frac{\exp(s_{\eta}(x,d^{+})/\tau_c)}{\exp(s_{\eta}(x,d^{+})/\tau_c) + \sum_{d^{-}\in\mathcal{N}^{-}} \exp(s_{\eta}(x,d^{-})/\tau_c)}.
\]
Large temperatures produce smoother distributions and broader supervision; small temperatures sharpen the objectives and focus training on top-ranked or hardest competing items. We then define the adaptive retriever objective as
\[
\mathcal{L}_{\mathrm{mix}}(\eta,\alpha_r,\alpha_c)
=
\mathcal{L}_{\mathrm{RAG}}(x,y;\tau_r)
+
\mathcal{L}_{\mathrm{InfoNCE}}(x,d^{+},\mathcal{N}^{-};\tau_c).
\]

\subsection{Gradient Interpretation of Adaptive Temperature}
\label{subsec:gradient_view}

Adaptive temperatures can be viewed as objective-side parameters that control how strongly the RAG and InfoNCE losses shape the query encoder. This explains how temperature learning acts as an implicit adaptive weighting mechanism, even though ARMOR does not attach explicit scalar mixture weights to the RAG and InfoNCE losses themselves.

We first consider the contrastive objective in a minimal two-document setting with one positive document \(d^{+}\), one negative document \(d^{-}\), and query embedding \(q\). Let \(s^{+}=q^{\top}d^{+}\) and \(s^{-}=q^{\top}d^{-}\), and define the margin \(\Delta=s^{+}-s^{-}\). The InfoNCE loss with temperature \(\tau_c\) is
\[
\mathcal{L}_{\mathrm{NCE}}(q,\tau_c)
=
\log\!\left(1+\exp(-\Delta/\tau_c)\right).
\]
Differentiating with respect to the query embedding yields
\[
\nabla_q \mathcal{L}_{\mathrm{NCE}}
=
-\frac{1}{\tau_c}
\sigma\!\left(-\frac{\Delta}{\tau_c}\right)
(d^{+}-d^{-}),
\]
where \(\sigma(\cdot)\) is the logistic sigmoid. Thus, \(\tau_c\) directly scales the gradient through \(1/\tau_c\) and controls how sharply updates focus on hard-margin examples. Smaller \(\tau_c\) makes the contrastive loss act more aggressively on the query encoder, while larger \(\tau_c\) smooths its influence.

An analogous effect appears in the retrieval component of the RAG objective. Consider a top-\(2\) setting with score gap \(\Delta=s_1-s_2\) between two retrieved documents. The temperature-modulated retrieval distribution is \(p_1 = \sigma(\Delta/\tau_r)\) and \(p_2 = 1-p_1\), giving the RAG loss
\[
\mathcal{L}_{\mathrm{RAG}}(q,\tau_r)
=
-\log\!\bigl(p_1 a_1 + p_2 a_2\bigr),
\]
where \(a_1\) and \(a_2\) denote the generator-side answer utilities of the two retrieved documents. Differentiating with respect to the score gap again yields an explicit \(1/\tau_r\) factor. Hence, smaller retrieval temperature makes the loss more sensitive to differences among top-ranked documents, while larger temperature spreads the training signal more broadly.

Together, these observations show that temperatures do more than change softmax entropy: they determine how strongly each objective contributes gradient signal to the query encoder. In this sense, \(\tau_r\) and \(\tau_c\) act as learned, objective-specific weighting mechanisms. Decreasing a temperature sharpens the corresponding objective and increases its local influence, while keeping it large makes the objective more conservative. This view explains how adaptive temperatures can replace manually tuned mixture weights while still allowing the relative influence of RAG likelihood and InfoNCE to evolve during training.

The main consequence of this formulation is that adaptive temperature learning provides a mechanism for \emph{balanced retriever optimization}. Rather than fixing the sharpness of the two objectives in advance, the model can begin training in a smoother regime and progressively sharpen one or both objectives as the retriever becomes more reliable. This makes the approach particularly appealing in low-resource settings, where fixed weighting can otherwise be brittle and highly sensitive to tuning.

At the same time, the same mechanism also introduces a potential failure mode. Once the retriever is already locally correct on many training examples, gradient descent can continue driving the temperatures downward, eventually making the softmaxes excessively sharp. In that regime, temperature learning no longer improves balance, but instead acts as a shortcut that amplifies existing score differences. This observation motivates the regularization strategy introduced below, which constrains the query encoder to remain compatible with the frozen base document space even as the retriever objectives become sharper.

\subsection{Query Distillation and Final ARMOR Objective}
\label{subsec:armor_objective}

Only the query encoder is updated in our setting, while the document encoder and index remain fixed in the base embedding space. As a result, improving the mixed retriever objective alone can still move the query encoder away from the geometry used at inference time. This issue is especially acute in low-resource domains, where optimization can over-specialize to the training objective.

To preserve compatibility with the frozen document space, we regularize the adapted query encoder toward the frozen base query encoder. For a minibatch of queries, we use cosine query distillation,
\[
\mathcal{L}_{\mathrm{qdist}}=
\frac{1}{B}\sum_{b=1}^{B}\left(1-\cos\!\bigl(q_{\eta}(x_b), q_0(x_b)\bigr)\right),
\]
where $q_{\eta}(x_b)$ is the adapted query embedding and $q_0(x_b)$ is the corresponding frozen base embedding. Because both query and document embeddings are $\ell_2$-normalized, cosine similarity provides a natural measure of deviation from the original retrieval geometry.

The adaptive temperatures and the query-distillation regularizer together define our full training framework, \textbf{ARMOR} (\textbf{A}daptive \textbf{R}egularized \textbf{M}ixture \textbf{O}ptimization for \textbf{R}etrievers). ARMOR combines adaptive mixture optimization of the RAG and InfoNCE objectives with regularization that preserves compatibility with the frozen base retriever space. 
The resulting objective is
\[
\mathcal{L}_{\mathrm{ARMOR}}(\eta,\alpha_r,\alpha_c)
=
\mathcal{L}_{\mathrm{RAG}}(x,y;\tau_r)
+
\mathcal{L}_{\mathrm{InfoNCE}}(x,d^{+},\mathcal{N}^{-};\tau_c)
+
\lambda_q\mathcal{L}_{\mathrm{qdist}}
\]
where $\lambda_q\geq 0$ controls the strength of query distillation. We choose $\lambda_q=1$ for ARMOR; setting $\lambda_q=0$ gives the dynamic-temperature variant without query regularization used in the ablation. Thus, the adaptive temperatures determine how selectively the utility-driven and geometry-driven objectives shape the query encoder, while $\lambda_q$ controls the compatibility penalty that prevents excessive drift away from the pretrained retrieval space.

\begin{figure}[!tbp]
\centering
    \includegraphics[width=\linewidth]{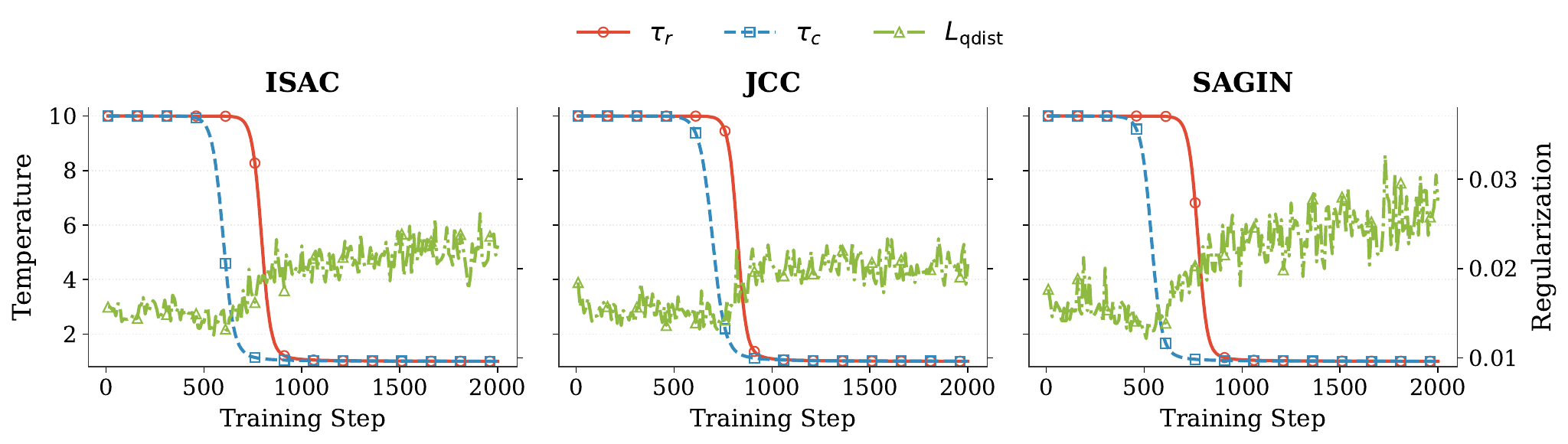}
    \caption{\textbf{Training dynamics of adaptive temperatures and query-distillation regularization across ISAC, JCC, and SAGIN domains.} Retrieval temperature consistently sharpens during training, while query-distillation loss rises late, indicating increasing tension between domain specialization and compatibility with the frozen embedding space.}
    \label{fig:training_dynamics_domains}
\end{figure}

\section{Related Work}

\noindent\textbf{Retrieval-Augmented Architectures.}
Retrieval-augmented systems use external evidence as non-parametric memory to complement model weights. REALM and RAG connected retrieval to downstream learning by treating documents as latent variables \cite{guu2020realm,lewis2020rag}, while FiD, Atlas, RETRO, and in-context retrieval-augmented language models established the benefits of modular memory, multi-passage reasoning, and updatable indices \cite{izacard2021fid,izacard2022atlas,borgeaud2022retro,asai2023ralm}. REPLUG is closest to our framing because it keeps the language model fixed and adapts prediction through retrieval \cite{shi2024replug}. However, REPLUG-style prediction can incur additional test-time overhead from document-wise scoring or aggregation, whereas our method retains the standard RAG-style retrieval-and-generation interface of \citet{lewis2020rag} while modifying the retriever-training objective to improve downstream performance.

\noindent\textbf{Dense Retrieval and Contrastive Learning.}
Dense passage retrieval introduced the dual-encoder paradigm for open-domain QA, embedding queries and documents in a shared space with contrastive supervision \cite{karpukhin2020dpr}. Contriever shows that such objectives can produce strong retrievers with limited labeled data, making them well suited to low-resource adaptation \cite{izacard2021contriever}. InfoNCE formalizes this discrimination-based training by contrasting positives and negatives, with temperature controlling softmax sharpness \cite{oord2018representation}. In contrast to purely contrastive retriever adaptation, we combine contrastive learning with generation-focused retrieval objectives and use adaptive temperature control to balance their influence during retriever fine-tuning.

\noindent\textbf{Telecom Grounding and Evaluation.}
Tele-Data and Tele-Eval provide telecom-specific corpora, QA supervision, and source identifiers for grounded QA \cite{maatouk2024tele}. KILT highlights the need to evaluate both answer correctness and evidence recovery in knowledge-intensive tasks \cite{petroni2021kilt}, while RAGAS offers reference-free measures of faithfulness and context relevance when exact citations are incomplete or noisy \cite{es2024ragas}. This is critical in telecom, where plausible but unsupported answers are insufficient. Building on these resources, we identify which components of a telecom RAG system should be adapted for domain-specific knowledge while limiting overfitting in low-resource settings.

\noindent\textbf{Regularization and Objective Balancing.}
Balancing multiple training signals is a recurring problem in multi-task and multi-objective learning. One line of work reweights task losses using uncertainty~\citep{kendall2018multitask}, gradient magnitudes~\citep{chen2018gradnorm}, or multi-objective formulations~\citep{sener2018multiobjective}. Another line works directly with gradients, aggregating or correcting them to reduce conflict and stochastic bias~\citep{yu2020gradient,liu2021conflict,fernando2022mitigating,fernando2024variance}. Penalty-based reformulations provide a related way to handle constrained or hierarchical optimization structure~\citep{shen2025mp}.

Similar trade-offs also appear in LLM post-training. Sequential SFT and preference learning can degrade safety, alignment, or earlier task performance~\citep{qi2023fine,lin2023speciality}, motivating methods that reformulate preference optimization, unify supervised and reinforcement-style objectives, or analyze and improve supervised--preference trade-offs~\citep{hong2024orpo,hua2024intuitive,chu2025sft,liu2025uft,fernando2024understanding}.

ARMOR is closest in spirit to adaptive loss weighting, but the adaptation is placed inside the retriever objectives rather than on external task coefficients. It learns temperatures for the RAG likelihood and InfoNCE softmax, then pairs this objective balancing with query distillation. This pairing is specific to the query-only RAG setting: because the document encoder and index remain frozen, regularization is needed to prevent over-sharpened retrieval and contrastive signals from moving the query encoder away from the fixed document embedding space.

\begin{figure}[!tbp]
\centering
\begin{subfigure}[t]{0.34\linewidth}
    \centering
    \includegraphics[width=\linewidth]{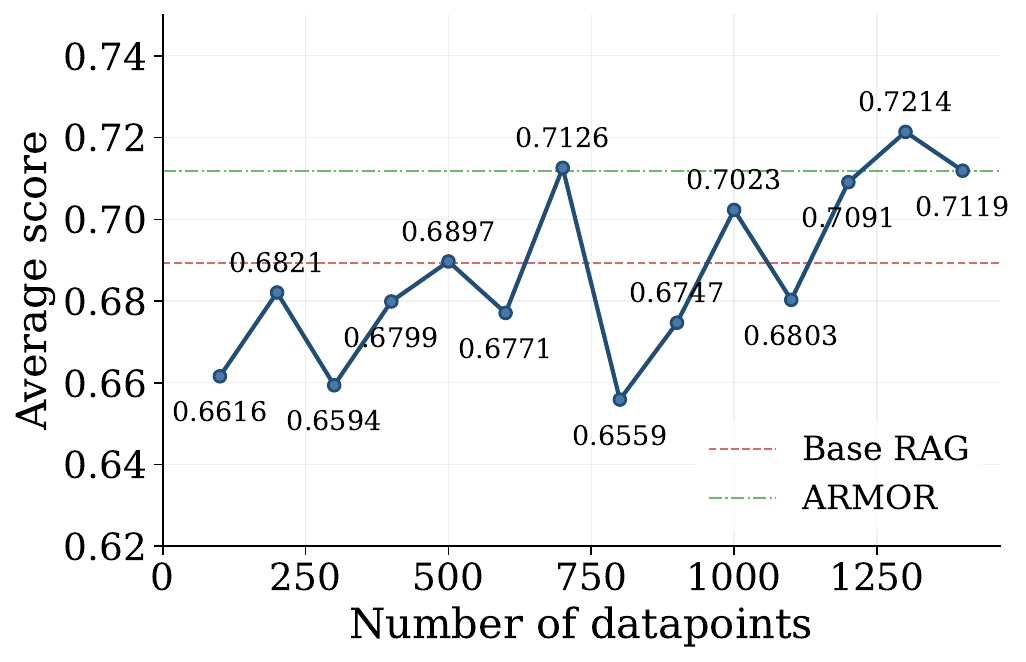}
    \caption{Training examples}
    \label{fig:datapoint_ablation}
\end{subfigure}
\hfill
\begin{subfigure}[t]{0.3\linewidth}
    \centering
    \includegraphics[width=\linewidth]{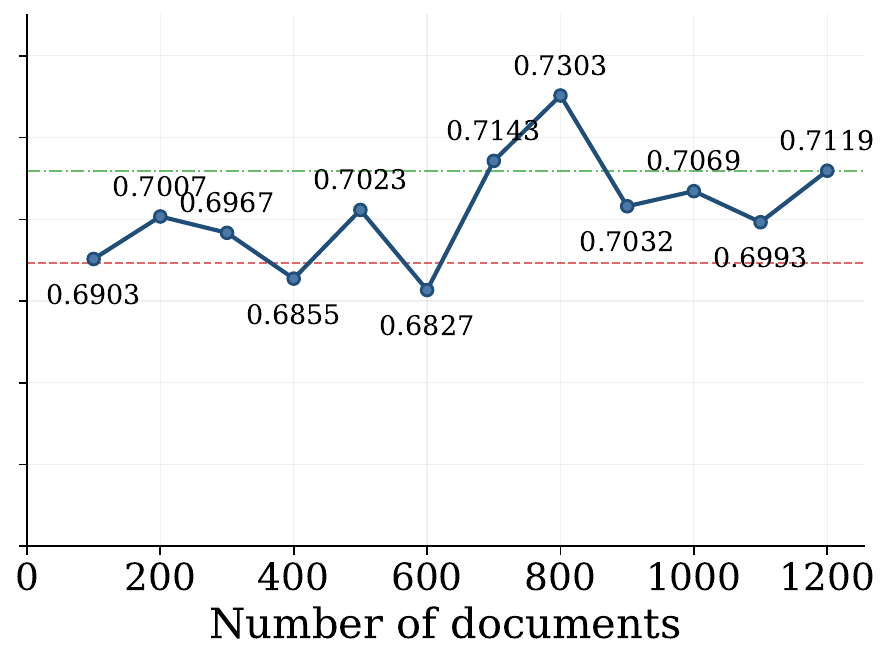}
    \caption{Source documents}
    \label{fig:document_ablation}
\end{subfigure}
\hfill
\begin{subfigure}[t]{0.28\linewidth}
    \centering
    \includegraphics[width=\linewidth]{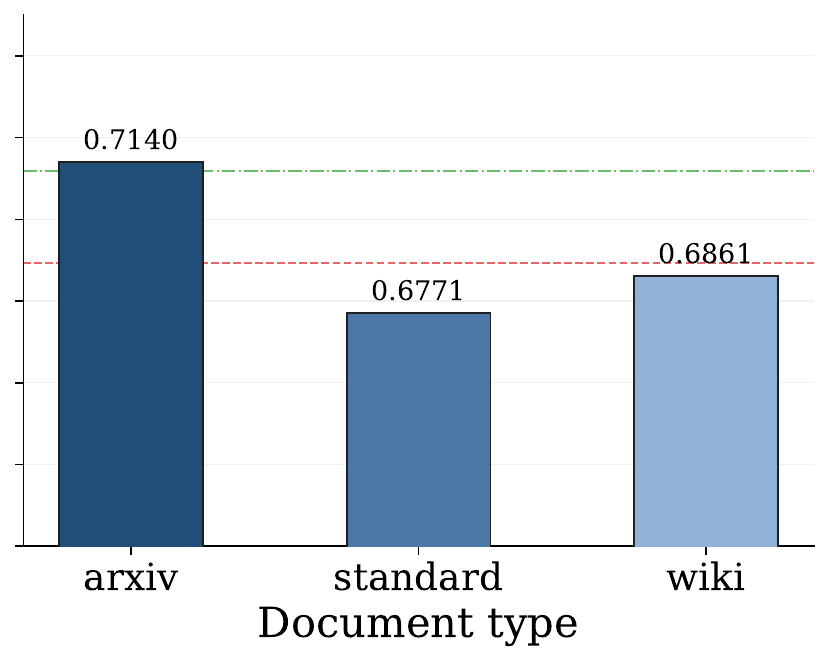}
    \caption{Document type}
    \label{fig:doctype_ablation}
\end{subfigure}
\caption{\textbf{ISAC Tele-Eval data and corpus ablations.} (a) and (b) vary ARMOR training supervision and source-document coverage; (c) restricts training to one source category. Performance improves overall but remains sensitive to source alignment and corpus composition.}
\label{fig:data_corpus_ablation}
\end{figure}

\section{Experiments}

We evaluate ARMOR on domain-specific telecom QA to understand when query-side retriever adaptation improves both evidence retrieval and downstream answer generation. The experiments first define a source-document-held-out Tele-Eval setting for in-domain evaluation, then examine main QA and retrieval performance, component ablations, ARMOR training dynamics, robustness across training data size and corpus composition, out-of-corpus TeleQnA evaluation, and generator scale.

\subsection{Experimental Setup}
\label{sec:exp_setup}

\noindent\textbf{Data and splits.}
The main experiments focus on three representative 6G subdomains: Integrated Sensing and Communication (ISAC), Joint Communication and Computation (JCC), and Space-Air-Ground Integrated Networks (SAGIN). We use two public resources: Tele-Eval~\citep{maatouk2024tele}, a domain-grounded telecom QA dataset, and Tele-Data~\citep{maatouk2024tele}, a large-scale telecom corpus containing standards, arXiv, and Wikipedia documents. For each retained Tele-Eval QA pair, the source document identifier is matched back to Tele-Data to construct a domain-specific retrieval corpus. Documents are split into 384-character passages and indexed once with the base \texttt{intfloat/e5-large-v2} encoder.

To reduce leakage, we split at the source-document level. Ten percent of unique source documents are reserved exclusively for testing; no QA pair derived from those documents appears in training or validation. Candidate chunks from each matched source document are then scored to identify the top three positive passages per question. These aligned positives provide contrastive supervision and define the ground-truth targets for Recall@\(k\). Table~\ref{tab:new_data_stats} summarizes the resulting data.

\begin{table}[t]
\centering
\small
\caption{\textbf{Tele-Eval experimental data summary.} Documents are matched to Tele-Data using Tele-Eval source identifiers and chunked into 384-character passages.}
\label{tab:new_data_stats}
\begin{tabular}{lccc}
\toprule
\textbf{Statistic} & \textbf{ISAC} & \textbf{JCC} & \textbf{SAGIN} \\
\midrule
Tele-Eval QA pairs after filtering & 1,701 & 3,060 & 1,339 \\
Unique source documents matched & 1,406 & 2,486 & 1,000 \\
Indexed chunks & 105,476 & 115,122 & 67,704 \\
Final train size & 1,361 & 2,466 & 1,048 \\
Final validation size & 178 & 294 & 137 \\
Final test size & 162 & 300 & 154 \\
\bottomrule
\end{tabular}
\end{table}

\noindent\textbf{Models and baselines.}
The dense retriever backbone is \texttt{intfloat/e5-large-v2}. Document embeddings are fixed in the base embedding space throughout training; only the query encoder is updated. The primary generator is Llama-3-8B-Instruct, with additional generator-scale comparisons using Llama-3.2-1B, Llama-3.2-3B~\citep{grattafiori2024llama}, and Qwen3-8B~\citep{qwen3technicalreport}. We compare \textbf{Base Gen} (closed-book generation), \textbf{Base RAG} (frozen retriever), \textbf{RAG QE FT} (query-encoder fine-tuning with RAG likelihood), \textbf{InfoNCE QE FT} (supervised contrastive fine-tuning), \textbf{Mix QE FT} (static mixture of RAG likelihood and InfoNCE), and \textbf{ARMOR} (adaptive temperatures plus query distillation).

\noindent\textbf{Evaluation.}
Tele-Eval measures open-ended QA and retrieval fidelity on the source-document-held-out test split. Unless otherwise noted, all reported Tele-Eval metrics use the first 150 examples from each domain's held-out test split. The retriever returns top-16 chunks, which are prepended to the prompt; the generator then produces a free-form answer. A GPT-5.2 judge scores the answer against the gold reference on a \(0\)--\(1\) scale, and retrieval quality is measured by Recall@1, Recall@3, and Recall@5 against the aligned positive chunks. Appendix~\ref{app:judge_prompts} lists the LLM judge prompts and key decoding parameters used in the experiments. TeleQnA~\citep{maatouk2025teleqna} provides an out-of-corpus multiple-choice robustness check. We report TeleQnA accuracy at top-\(k=16\), matching the retrieval breadth used for Tele-Eval generation.

\subsection{Experiment Results}\label{sec:results}

\begin{table*}[t]
\centering
\footnotesize
\setlength{\tabcolsep}{3pt}
\renewcommand{\arraystretch}{0.95}
\caption{\textbf{Tele-Eval open-ended QA and retrieval results across ISAC, JCC, and SAGIN domains.} Score is the GPT-5.2 judge average answer score; R@\(k\) is recall against aligned positive passages. Best values within each domain and metric are bolded.}
\label{tab:tele_eval}
\begin{tabular}{lcccccccccccc}
\toprule
\multirow{3}{*}{Method}
& \multicolumn{12}{c}{Tele-Eval, top-16 retrieval} \\
\cmidrule(lr){2-13}
& \multicolumn{4}{c}{ISAC}
& \multicolumn{4}{c}{JCC}
& \multicolumn{4}{c}{SAGIN} \\
\cmidrule(lr){2-5} \cmidrule(lr){6-9} \cmidrule(lr){10-13}
& Score & R@1 & R@3 & R@5
& Score & R@1 & R@3 & R@5
& Score & R@1 & R@3 & R@5 \\
\midrule
Base Gen
& 0.2269 & -- & -- & --
& 0.2980 & -- & -- & --
& 0.3017 & -- & -- & -- \\
Base RAG
& 0.6893 & \textbf{0.5467} & 0.7400 & 0.8067
& \textbf{0.7763} & 0.4800 & 0.6133 & 0.7000
& 0.7660 & \textbf{0.6400} & 0.8133 & 0.8400 \\
RAG QE FT
& 0.6584 & 0.5000 & 0.6533 & 0.7400
& 0.7230 & 0.3867 & 0.6000 & 0.6533
& 0.7573 & 0.5000 & 0.6933 & 0.7867 \\
InfoNCE QE FT
& 0.6685 & 0.5200 & 0.6733 & 0.7533
& 0.7425 & 0.4133 & 0.6200 & 0.6800
& 0.7662 & 0.5333 & 0.7400 & 0.8067 \\
Mix QE FT
& 0.6854 & 0.4733 & 0.6333 & 0.7133
& 0.7360 & 0.4467 & 0.6400 & 0.6800
& 0.7591 & 0.5000 & 0.7200 & 0.7800 \\
ARMOR
& \textbf{0.7119} & 0.5267 & \textbf{0.7667} & \textbf{0.8200}
& 0.7719 & \textbf{0.4933} & \textbf{0.6467} & \textbf{0.7133}
& \textbf{0.7685} & 0.6267 & \textbf{0.8400} & \textbf{0.8467} \\
\bottomrule
\end{tabular}
\end{table*}

\noindent\textbf{Main Tele-Eval Results.} Table~\ref{tab:tele_eval} shows that the harder source-document-split setting changes the empirical picture. Every non-ARMOR adaptation method degrades below Base RAG on at least one domain's answer score, whereas ARMOR is the only fine-tuned method that consistently matches or improves the frozen baseline. On ISAC, ARMOR improves the answer score from 0.6893 to 0.7119 while also achieving the best Recall@3 and Recall@5. On SAGIN, ARMOR gives the strongest answer score and improves slightly over Base RAG. On JCC, where Base RAG is already strongest, ARMOR approximately preserves answer quality while giving the best retrieval recall at all reported ranks.

The retrieval pattern is especially informative: ARMOR dominates R@3 and R@5 in every domain, even where Base RAG keeps the best R@1. In the top-16 generation setting, downstream answer quality depends on whether useful evidence appears in the context window, not only on whether the single highest-ranked passage is positive. ARMOR therefore appears to improve evidence coverage without collapsing retrieval mass onto a brittle top-1 decision.

\begin{table}[!tbp]
\centering
\caption{\textbf{ARMOR component ablation with ISAC Tele-Eval.} Adaptive temperatures and query-distillation regularization are complementary: temperature learning without regularization performs worst, while full ARMOR performs best.}
\label{tab:component_ablation}
\begin{tabular}{lccc}
\toprule
\textbf{Method} & \textbf{Adaptive Temps.} & \textbf{Regularization} & \textbf{Avg. Score} \\
\midrule
Base RAG & -- & -- & 0.6893 \\
RAG QE FT & -- & -- & 0.6584 \\
InfoNCE QE FT & -- & -- & 0.6685 \\
Mix QE FT & no & no & 0.6854 \\
Static Mix with Reg. & no & yes & 0.6729 \\
Dynamic Mix without Reg. & yes & no & 0.6350 \\
ARMOR & yes & yes & \textbf{0.7119} \\
\bottomrule
\end{tabular}
\end{table}

\noindent\textbf{ARMOR Component Ablation.} Table~\ref{tab:component_ablation} isolates the contribution of ARMOR's two design choices. Dynamic Mix without Reg. performs worse than all other fine-tuned methods, which supports the central failure mode: unconstrained temperature learning can over-sharpen retrieval distributions and pull query embeddings away from the frozen document space. Regularization alone does not recover the full gain either, since Static Mix with Reg. remains below the unregularized static mixture. Only the combination of adaptive temperatures and query distillation improves substantially over both single-objective and static-mixture baselines.

\noindent\textbf{Data and Corpus Ablations.} Figure~\ref{fig:data_corpus_ablation} studies how ARMOR's performance changes as the amount of ISAC training supervision, source-document coverage, and document-type composition vary. Panels (a) and (b) trend upward overall, but neither is perfectly monotonic. This is expected in a source-document-split setting: adding examples or documents changes the balance among RAG likelihood, InfoNCE, and query distillation, and intermediate subsets can be less aligned with the held-out evaluation set than smaller or larger subsets. Panel (c) provides a complementary view by restricting training to one document category, reinforcing that retriever adaptation depends not only on how much data is available, but also on whether the document type matches the evidence needed at evaluation time.

\noindent\textbf{Training Dynamics.} Figure~\ref{fig:training_dynamics_domains} shows that the learned temperatures evolve on different timescales. Retrieval-side sharpening is robust across domains, while contrastive-side sharpening is more domain-dependent. The late rise in query-distillation loss is consistent with the ablation in Table~\ref{tab:component_ablation}: adaptive temperatures are useful because they allow the influence of the two objectives to change during training, but they require regularization to prevent harmful drift from the frozen document index.

\noindent\textbf{Out-of-Corpus Transfer via TeleQnA.} TeleQnA plays a different role from Tele-Eval. Because it is multiple-choice, the model often needs only enough topical signal to eliminate wrong options, so fine-grained retrieval differences can collapse into similar accuracy values. Table~\ref{tab:tele_qna} shows that ARMOR remains competitive at top-\(k=16\) without catastrophic out-of-corpus degradation: it is within two points of the best result on ISAC and JCC and ties the best result on SAGIN.

\begin{figure}[t]
    \centering
    \includegraphics[width=0.7\linewidth]{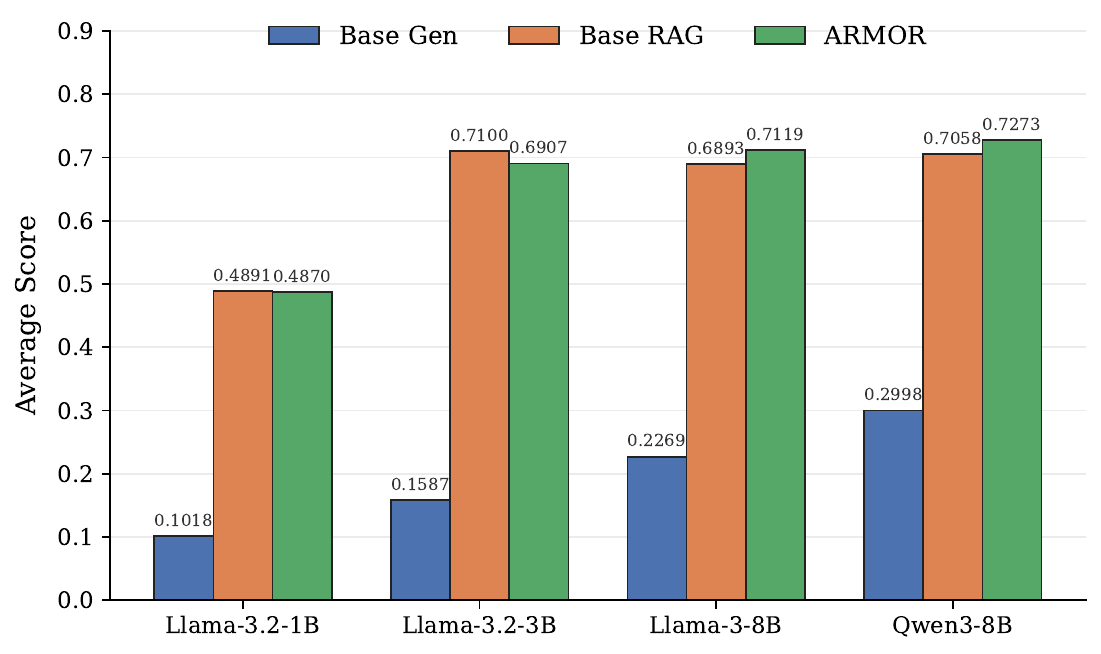}
    \caption{\textbf{Comparison of Base Gen, Base RAG, and ARMOR across generator backbones on ISAC Tele-Eval.} ARMOR's gains are clearest for 8B-scale generators, suggesting that stronger generators are better able to use improved retrieved evidence.}
    \label{fig:model_performance_comparison}
\end{figure}

\begin{table}[t]
\centering
\caption{\textbf{TeleQnA multiple-choice accuracy.} TeleQnA is used as an out-of-corpus robustness check rather than the primary retrieval-sensitive benchmark.}
\label{tab:tele_qna}
\begin{tabular}{lccc}
\toprule
\textbf{Method} & \textbf{ISAC} & \textbf{JCC} & \textbf{SAGIN} \\
\midrule
Base Gen & 0.7000 & 0.6000 & 0.6800 \\
Base RAG & 0.8200 & \textbf{0.7600} & 0.7600 \\
RAG QE FT & 0.8400 & 0.7200 & 0.7800 \\
InfoNCE QE FT & 0.8400 & 0.7200 & \textbf{0.8200} \\
Mix QE FT & \textbf{0.8800} & 0.7400 & \textbf{0.8200} \\
ARMOR & 0.8600 & 0.7400 & \textbf{0.8200} \\
\bottomrule
\end{tabular}
\end{table}

\noindent\textbf{Effect of Generator Scale.} Figure~\ref{fig:model_performance_comparison} shows that ARMOR's benefit over Base RAG grows with generator capacity rather than shrinking. For Llama-3.2-1B, ARMOR and Base RAG are essentially tied; for Llama-3.2-3B, ARMOR is slightly below Base RAG; for both 8B models, ARMOR outperforms Base RAG by roughly two to three points. This suggests that retriever optimization is most valuable when the generator is capable enough to synthesize and use a better evidence set. Qwen3-8B also outperforms Llama-3-8B under both Base RAG and ARMOR, indicating that model family matters in addition to parameter count.

\section{Conclusion}

We studied low-resource telecom QA and showed that retriever-side adaptation can be a more effective and operationally stable alternative to generator fine-tuning. We proposed \textbf{ARMOR}, which combines RAG likelihood and InfoNCE through adaptive temperature-based objective shaping, while using query distillation to preserve compatibility with the frozen document embedding space. Empirically, ARMOR provides a robust tradeoff between retrieval specialization and downstream QA performance, with training dynamics showing that adaptive temperatures sharpen the retrieval signal while regularization limits query-encoder drift. These results suggest that low-resource RAG adaptation should jointly optimize complementary retriever objectives under compatibility constraints. 

\bibliographystyle{plainnat}
\bibliography{bib}

\newpage
\appendix

\section{Experiment Details} \label{app:experiments}

In this section, we provide additional details on the data generation, baseline implementation, and experiment setup used in this paper. All experiments were conducted on a server equipped with 8 NVIDIA A100 80GB SXM GPUs.

\subsection{Data Sources and Experiment Pipelines}
\label{app:data_generation}

This paper uses two related but distinct data pipelines. The introductory comparison in Figure~\ref{fig:intro_objective_tradeoffs} is a motivation experiment based on an earlier standards-only data-generation pipeline. The main experiments in Section~\ref{sec:exp_setup} use the newer Tele-Eval source-document-split setup summarized in Table~\ref{tab:new_data_stats}. We separate them here to avoid conflating the motivation figure with the final evaluation protocol.

\noindent\textbf{Introductory motivation pipeline.}
For Figure~\ref{fig:intro_objective_tradeoffs}, we curated domain-specific training data from the standards portion of Tele-Data~\citep{maatouk2024tele}. Documents were filtered by domain-specific keywords and LLM refinement, then chunked with the \texttt{intfloat/e5-large-v2} tokenizer. Grounded QA pairs were generated from the retained chunks using an instruction-tuned generator, and positive retrieval targets were aligned back to source chunks using document identifiers and lexical-overlap signals. This pipeline was used only to motivate the choice of optimizing the query-side retriever rather than the generator.

\noindent\textbf{Main experimental pipeline.}
The main experiments use Tele-Eval QA pairs and match each retained question to its source document in Tele-Data. The retained document pool spans standards, arXiv, and Wikipedia sources. After domain filtering, each matched source document is split into 384-character passages and indexed with FAISS using base \texttt{intfloat/e5-large-v2} embeddings. The document index is built once and remains fixed during all retriever training.

\noindent\textbf{Source-document split.}
To prevent leakage, the main experiments hold out 10\% of unique source document IDs for testing. All QA pairs derived from those documents are assigned to the test split, while the remaining QA pairs are divided into training and validation sets. This policy evaluates whether the adapted query encoder generalizes to unseen technical documents rather than merely memorizing passages from documents seen during training.

\noindent\textbf{Positive passage alignment.}
For each retained QA pair, candidate chunks from the matched source document are scored to select the top three positive passages. These positives are used as supervised targets for InfoNCE training and as the ground-truth evidence set for Recall@1, Recall@3, and Recall@5. InfoNCE negatives are drawn from retrieved chunks outside the positive set together with in-batch negatives.

\subsection{Baseline Retriever Training Implementations}
\label{app:baseline_retriever_training}

All retriever-training methods initialize from \texttt{intfloat/e5-large-v2}; the document encoder and document index are fixed, and only the query encoder is updated. The baselines are designed to isolate which retriever objective and regularization choices are responsible for downstream QA behavior.

\noindent\textbf{RAG QE FT.}
This baseline optimizes the query encoder with the RAG negative log-likelihood objective under a frozen generator:
\begin{equation}
\mathcal{L}_{\mathrm{RAG}}
=
-\log \sum_{d \in D}
p_{\eta}(d \mid x)
p_{\theta_0}(y \mid x, d).
\end{equation}
The retriever is trained to assign higher probability to chunks that improve likelihood of the target answer.

\noindent\textbf{InfoNCE QE FT.}
This baseline uses supervised contrastive training with the aligned positive chunks and dynamically selected negatives. It directly improves separation between query embeddings and aligned evidence chunks.

\noindent\textbf{Mix QE FT.}
This baseline uses a fixed-weight mixture of the RAG and InfoNCE losses:
\begin{equation}
\mathcal{L}_{\mathrm{mix}}
=
\lambda_{\mathrm{RAG}} \mathcal{L}_{\mathrm{RAG}}
+
\lambda_{\mathrm{cont}} \mathcal{L}_{\mathrm{cont}}.
\end{equation}
It tests whether simply combining the two objectives is sufficient without adaptive temperatures or query regularization.

\noindent\textbf{ARMOR variants.}
The component ablation in Table~\ref{tab:component_ablation} separates the effects of adaptive temperatures and query-distillation regularization. Static Mix with Reg. adds query distillation to the static mixture while keeping fixed objective weights. Dynamic Mix without Reg. learns the unconstrained temperature parameters $\alpha_r$ and $\alpha_c$ but sets $\lambda_q=0$, thereby removing query distillation. Full ARMOR learns the same temperature parameters and uses $\lambda_q=1$ for query distillation.

\subsection{Models}
\label{app:models_used}

For the main experiments, the primary generator is Llama-3-8B-Instruct. Generator-scale analysis additionally evaluates Llama-3.2-1B, Llama-3.2-3B~\citep{grattafiori2024llama}, and Qwen3-8B~\citep{qwen3technicalreport}. Across all runs, the dense retriever backbone is \texttt{intfloat/e5-large-v2}. Document embeddings are always computed with the frozen base encoder. Query embeddings are initialized from the same encoder and then adapted by the corresponding retriever-training objective.

\subsection{Evaluation Harness and Methodology}
\label{app:evaluation_harness}

\noindent\textbf{Tele-Eval.}
Tele-Eval is the primary benchmark because it measures open-ended answer generation and evidence retrieval. For each method, the retriever returns top-16 chunks from the fixed FAISS index. The generator conditions on these chunks and produces a free-form answer. GPT-5.2 then grades the generated answer against the gold reference on a continuous \([0,1]\) scale. We report the mean score over the first 150 test examples in each domain. Retrieval quality is measured by Recall@1, Recall@3, and Recall@5 against the aligned positive chunks, again on the first 150 test examples.

\noindent\textbf{TeleQnA.}
TeleQnA~\citep{maatouk2025teleqna} is used as an out-of-corpus multiple-choice robustness check. Domain-specific TeleQnA subsets are constructed by matching questions to ISAC, JCC, and SAGIN profiles. In RAG settings, the retriever supplies top-16 chunks, and the generator selects one answer option. Accuracy is the fraction of examples for which the parsed option matches the corrected gold option.

\noindent\textbf{Interpretation.}
Tele-Eval is more sensitive to retrieval quality because the model must synthesize a grounded free-form answer from the retrieved evidence. TeleQnA is less retrieval-sensitive because a topically relevant context can often be enough to eliminate incorrect choices. We therefore treat Tele-Eval as the main benchmark and TeleQnA as a transfer and robustness check.

\subsection{LLM Judge Prompts and Parameters}
\label{app:judge_prompts}

We used LLM judges in four parts of the data and evaluation pipeline. The prompts below are grouped in the order in which the corresponding procedures are discussed in the paper. Unless otherwise stated, OpenAI-based judges used temperature $0.0$, parsed the first valid numeric score from the model output, and clamped numeric scores to $[0,1]$.

\paragraph{Prompts used for the Tele-Data standards data-generation pipeline.}
These prompts were used in the earlier standards-only Tele-Data pipeline that supports the introductory motivation experiment. They are separate from the main Tele-Eval source-document-split experiments.

\noindent\textbf{OpenAI Document Relevance Post-Filter.}
This judge post-filters cleaned Tele-Data standards documents before example generation. The default OpenAI model was GPT-5.2, the temperature was $0.0$, the relevance threshold was $0.50$, the snippet length was 2000 characters, the timeout was 120 seconds, and the maximum number of retries was 6.

\begin{promptbox}{OpenAI Document Relevance Post-Filter: system prompt}
You are a wireless communications and sensing domain expert curating a tutorial dataset.
\end{promptbox}

\begin{promptbox}{OpenAI Document Relevance Post-Filter: user prompt}
Decide whether this Tele-Data sample is suitable for a tutorial dataset focused on {domain_name}.
{domain_specific_instructions}

Return ONLY a single number between 0.0 and 1.0 (no JSON, no words):
- 1.0 = clear true positive for a {domain_name} tutorial dataset
- 0.0 = clear false positive

SAMPLE_ID: {sid}
CATEGORY: {cat}
META_HINT: {meta_hint}
CONTENT_SNIPPET:
<<<
{snippet}
>>>
\end{promptbox}

For the ISAC data-generation run, the domain-specific instructions asked the judge to include explanatory or instructional content on ISAC/JCAS fundamentals, joint waveform and signal design, sensing tasks in communication systems, joint beamforming and resource allocation, sensing-communication performance tradeoffs, and standards-oriented discussion relevant to cellular or wireless ISAC. The instructions excluded marketing material, vendor promotions, business or market analysis, and speculative vision papers without technical explanation.

\noindent\textbf{Generated Example Quality Judge.}
After local-vLLM example generation, this judge determines whether each generated training example should be retained. The default model was \texttt{meta-llama/Llama-3.3-70B-Instruct}, with sampling temperature $0.2$, top-$p$ $0.9$, maximum output length 900 tokens, judge batch size 20, and tensor parallel size 2. Retained examples required \texttt{keep=true}, \texttt{answerable\_without\_context=true}, technical score at least $0.70$, and clarity score at least $0.70$.

\begin{promptbox}{Generated Example Quality Judge: system prompt}
You are a strict dataset quality reviewer for ISAC tutorial data.
\end{promptbox}

\begin{promptbox}{Generated Example Quality Judge: user prompt}
Judge whether the following generated training example should be kept.

Reject if:
- It requires the original source text to understand (context-dependent, refers to 'this paper/section/figure/table').
- It is vague, hand-wavy, or mainly fluff.
- It contains factual errors or mismatched concepts (e.g., confusing radar Doppler with carrier frequency offset).
- It is not about ISAC/JCAS or has no sensing+comm linkage.
- It includes citations, section numbers, or local references.

Return ONLY JSON with this exact shape:
{
  "keep": true/false,
  "answerable_without_context": true/false,
  "technical_score": 0.0-1.0,
  "clarity_score": 0.0-1.0,
  "issues": ["..."]
}

EXAMPLE_TO_JUDGE_JSON:
<<<
{example_json}
>>>
\end{promptbox}

\paragraph{Prompts used for filtering domain-specific Tele-Eval data.}
These prompts were used in the main Tele-Eval source-document-split pipeline to filter QA pairs by target domain and align each retained QA pair with supporting chunks from its source document.

\noindent\textbf{Tele-Eval QA Domain Filtering Judge.}
This judge assigns each Tele-Eval QA pair domain relevance scores for ISAC, SAGIN, and JCC. The default model was GPT-5.2, the temperature was $0.0$, the keep threshold was $0.65$, the timeout was 60 seconds, and the maximum number of retries was 3. If the OpenAI API key was unavailable in OpenAI mode, the pipeline fell back to keyword scoring.

\begin{promptbox}{Tele-Eval QA Domain Filtering Judge: system prompt}
You are a strict telecom dataset curator. Return only three comma-separated numeric scores.
\end{promptbox}

\begin{promptbox}{Tele-Eval QA Domain Filtering Judge: user prompt}
Judge whether this telecom QA pair belongs in a demo corpus focused on these domains:
{domain_lines}

Return ONLY three numbers between 0.0 and 1.0, in this exact order:
isac,sagin,jcc

Format requirements:
0.0,0.0,0.0
Do not return JSON.
Do not return prose.
Do not return labels, markdown, explanations, or units.

Each score is the QA pair's relevance to that specific domain.
For each domain, 1.0 means clearly technical, answerable, and relevant to that domain.
For each domain, 0.0 means off-topic for that domain, generic, non-technical, malformed, or lacking an answer.

QUESTION:
{question}

ANSWER:
{answer}
\end{promptbox}

\noindent\textbf{Contriever Alignment Chunk Judge.}
This judge scores candidate chunks from the matched source document and selects the top positive passages used for InfoNCE training and Recall@$k$ evaluation. The default model was GPT-5.2, the temperature was $0.0$, the judge batch size was 8, the top positive count was 3, the snippet length was 1800 characters, the timeout was 60 seconds, and the maximum number of retries was 3. This chunk-alignment judge was not needed in the earlier generated-data pipeline because each QA pair was generated from a known source chunk. In the Tele-Eval pipeline, the source document is known, but the supporting chunks must still be aligned for supervised retriever training and retrieval evaluation.

\begin{promptbox}{Contriever Alignment Chunk Judge: system prompt}
You are a strict retrieval alignment judge. Return only comma-separated numeric scores.
\end{promptbox}

\begin{promptbox}{Contriever Alignment Chunk Judge: user prompt}
You are judging which document chunks were most likely used to generate a question-answer pair.

Return ONLY {len(chunks)} comma-separated numbers between 0.0 and 1.0.
Do not return JSON, labels, markdown, prose, explanations, or units.

Score each chunk independently:
- 1.0 = the chunk directly supports the answer to the question
- 0.0 = the chunk is unrelated or insufficient

QUESTION:
{question}

ANSWER:
{answer}

CHUNKS:
{chunk_blocks}
\end{promptbox}

Each chunk block was formatted as \texttt{CHUNK \{idx\} | chunk\_id=\{chunk\_id\} | vid=\{vid\}} followed by the shortened chunk text.

\paragraph{Prompts used for filtering domain-specific TeleQnA questions.}
The domain-filtered TeleQnA splits used an optional LLM judge after keyword/domain filtering. This judge scored whether a multiple-choice QA item was substantively relevant to the target domain. The default model was GPT-5.2, the temperature was $0.0$, the maximum output length was 16 tokens, and the maximum field length was 5000 characters. The confidence threshold was $0.70$ in the original wrapper and $0.75$ in the ARMOR 2.0 split-generation wrapper.

\begin{promptbox}{TeleQnA Domain Relevance Confidence Judge: system prompt}
You are a strict wireless communications domain relevance classifier. You output only one numeric confidence score.
\end{promptbox}

\begin{promptbox}{TeleQnA Domain Relevance Confidence Judge: user prompt}
Decide whether the following TeleQnA multiple-choice Q/A pair is relevant to the domain: {domain_display_name}.

Target domain includes:
{domain_include}

Target domain excludes:
{domain_exclude}

Scoring:
- 1.0 = clearly and substantively relevant to the target domain.
- 0.0 = not relevant to the target domain.
- Use values between 0.0 and 1.0 for partial or uncertain relevance.
- Prefer precision over recall.
- Do not give high scores for generic telecom terms unless there is a clear domain connection.

Return ONLY one floating-point number between 0 and 1.
No explanation. No JSON. No text.

QUESTION:
<<<
{question}
>>>

OPTIONS:
<<<
{options}
>>>

GOLD ANSWER:
<<<
{answer}
>>>

EXPLANATION:
<<<
{explanation}
>>>
\end{promptbox}

\paragraph{Prompt used for evaluating open-ended Tele-Eval answers.}
This judge evaluates generated open-ended answers by comparing each candidate response against the reference answer and returning a scalar score. The default model was GPT-5.2, the temperature was $0.0$, the maximum completion length was 32 tokens, and failures returned score $0.0$.

\begin{promptbox}{Tele-Eval Answer Grading Judge: system prompt}
You are an automatic grader for question-answer pairs. You output only one numeric score.
\end{promptbox}

\begin{promptbox}{Tele-Eval Answer Grading Judge: user prompt}
You are grading a model answer for a question-answer pair.

Compare the candidate answer to the ground-truth reference.

Score semantics:
- 1.0 = on par with or better than the reference answer (technically correct, covers the essential idea).
- <1.0 = worse than the reference.

Return ONLY one floating-point number between 0 and 1. No explanation, no text, no JSON.

Question:
<<<
{prompt}
>>>

Ground-truth answer:
<<<
{reference}
>>>

Candidate model answer:
<<<
{candidate}
>>>
\end{promptbox}

\section{Capacity comparison for SFT and RAG fine-tuning}\label{sec:theory_proof}
In this section, we derive generalization bounds for generator and retriever fine-tuning as a capacity comparison. The purpose is to identify one estimation-complexity reason retriever adaptation can be attractive in low-data domains, rather than to characterize all regimes in which retriever tuning should outperform generator tuning.

\subsection{General setup}
We compare two families of adaptation, \textit{SFT} and \textit{RAG}, on the same population distribution $\mathcal{P}$ over examples $z=(x,y)$, where $x$ is the question, and $y$ is the target answer. Given $N$ i.i.d. training pairs $\{z_i\}_{i=1}^N$ drawn from the population distribution $\mathcal{P}$, let $\theta\in\mathcal{E}_S\subset\mathbb{R}^{d_S}$ be the trainable model parameter for SFT, $\eta\in\mathcal{E}_R\subset\mathbb{R}^{d_R}$ be the trainable RAG parameters, and $\ell_S(z ; \theta)$ and $\ell_R(z ; \eta)$ be the per-sample SFT and RAG loss, respectively.

\subsubsection{Population and empirical losses}
The population SFT loss is defined as the expected test loss under the true population distribution $\mathcal{P}$, while the empirical SFT loss is defined as the empirical training loss on finite sample $\{z_i\}_{i=1}^N$, i.e. 
\begin{align*}
\textbf{population loss: } L_S(\theta)=\mathbb{E}_{z \sim \mathcal{P}}\left[\ell_S(z ; \theta)\right], ~~ \textbf{empirical loss: } \widehat{L}_{S, N}(\theta)=\frac{1}{N} \sum_{i=1}^N \ell_S(z_i ; \theta).
\end{align*}
The corresponding population and empirical risk SFT minimizer are defined as 
\begin{align*}
&\textbf{population solution: } \theta^* \in \arg \min _{\theta \in \mathcal{E}_S} L_S(\theta), \quad \textbf{empirical solution: } \widehat{\theta} \in \arg \min _{\theta \in \mathcal{E}_S} \widehat{L}_{S, N}(\theta) .
\end{align*}

Similarly, we can define the population and empirical RAG losses and their corresponding optimal solutions as follows. 
\begin{align*}
&\textbf{population loss: } L_R(\eta)=\mathbb{E}_{z \sim \mathcal{P}}\left[\ell_R(z ; \eta)\right], ~~ \textbf{empirical loss: } \widehat{L}_{R, N}(\eta)=\frac{1}{N} \sum_{i=1}^N \ell_R(z_i ; \eta). \\
&\textbf{population solution: } \eta^* \in \arg \min _{\eta \in \mathcal{E}_R} L_R(\eta), \quad \textbf{empirical solution: } \widehat{\eta} \in \arg \min _{\eta \in \mathcal{E}_R} \widehat{L}_{R, N}(\eta) .
\end{align*}

A standard goal of generalization theory is to control the gap between the population loss and the empirical loss. In our comparison, a smaller upper bound should be interpreted as better estimation control under the stated assumptions rather than as a direct prediction of downstream performance.

\subsubsection{Function class and Rademacher complexity}
To state the comparison, we rewrite the objectives for SFT and RAG in terms of their backbone prediction model class. 

Let $\mathcal{Y}$ denote the answer label space and denote the RAG answer distribution as
\begin{align*}
q_\eta(\cdot ~|~ x)=\sum_{d \in \mathcal{N}_k(x)} p_\eta(d \mid x) p_{\theta_0}(\cdot \mid x, d)
\end{align*}
where $\mathcal{N}_k(x)$ is the top-$k$ document set ranked by the retrieval score $s_\eta (x,d)=f_\eta(x)^\top g(d)$, and
\[
p_\eta(d \mid x)=\frac{\exp \left(s_\eta(x, d)\right)}{\sum_{d^{\prime} \in \mathcal{N}_k(x)} \exp \left(s_\eta\left(x, d^{\prime}\right)\right)}.
\]
Therefore, we can write the backbone function class for SFT and RAG as 
\begin{align}\label{SFT_RAG_function_class}
\mathcal{F}_\theta=\left\{x \mapsto p_\theta(\cdot ~|~ x): \theta \in \mathcal{E}_S\right\}, \quad \mathcal{F}_\eta=\left\{x \mapsto q_\eta (\cdot ~|~ x): \eta \in \mathcal{E}_R\right\}. 
\end{align}
Letting $\ell_{\operatorname{NLL}}(q, x,y)=-\log q(y~|~x)$ be the negative log-likelihood loss, the SFT and RAG loss can be rewritten as 
\begin{align*}
\ell_S(z ; \theta)=\ell_{\mathrm{NLL}}(p_\theta, x, y), \text{ and } \ell_R(z ; \eta)=\ell_{\mathrm{NLL}}(q_\eta, x, y).
\end{align*}

\noindent\textbf{Rademacher complexity.} Given the training samples $\mathcal{Z}_N=\{z_i\}_{n=1}^N$, the empirical Rademacher complexity of a function class $\mathcal{F}$ for these samples is defined as follows 
\begin{align*}
\mathcal{R}_{\mathcal{Z}_N}(\mathcal{F})=\frac{1}{N} \mathbb{E}_{\epsilon \sim \operatorname{unif}(\{1,-1\})}\left[\sup _{f \in \mathcal{F}} \sum_{n=1}^N \epsilon_n f (x_n )\right]
\end{align*}

With Rademacher complexity, we have the following standard generalization bound. 

\begin{theorem}\label{thm:generalization}
Assume that the loss function $\ell_{\mathrm{NLL}}(q,x,y)$ is bounded in $[0, c]$ and is $\rho$-Lipschitz with respect to $q$ in the feasible domain. Then, with probability at least $1-\delta$ over the samples $\mathcal{Z}_N=\left\{z_n\right\}_{n=1}^N$, we have
\begin{align}
&\sup _{\theta \in \mathcal{E}_S}\left\{L_S(\theta)-\widehat L_{S,N}(\theta)\right\} \leq 2 \rho \mathcal{R}_{\mathcal{Z}_N}(\mathcal{F}_\theta)+3 c \sqrt{\frac{\log (2 / \delta)}{2 N}}\\
&\sup _{\eta \in \mathcal{E}_R}\left\{L_R(\eta)-\widehat L_{R,N}(\eta)\right\} \leq 2 \rho \mathcal{R}_{\mathcal{Z}_N}(\mathcal{F}_\eta)+3 c \sqrt{\frac{\log (2 / \delta)}{2 N}}
\end{align}
\end{theorem}

\begin{proof}
The proof follows directly from the generalization bounds for machine learning models in terms of Rademacher complexity, e.g. \citep[Theorem 3.1]{mohri2018foundations}, or \citep[Theorem B.1]{arora2019fine}. 
\end{proof}

According to Theorem \ref{thm:generalization}, the second terms in the generalization bounds for SFT and RAG are the same, so we will compare the Rademacher complexity $\mathcal{R}_{\mathcal{Z}}(\mathcal{F}_\theta)$ for SFT and $\mathcal{R}_{\mathcal{Z}}(\mathcal{F}_\eta)$ for RAG in the subsequent sections given the particular model architecture.

\subsection{Generalization bound for generator fine-tuning}

In this section, we specify the generator model as a Transformer-based model. For simplicity, we analyze the one-layer Transformer model following \citep{mwigogeneralization}, but extension to multi-layer Transformer is also possible \citep{trauger2024sequence}. We denote the input for a single question-answer pair $(x,y)$ after tokenization as $X\in\mathbb{R}^{d_S\times d}$, and write the SFT parameter $\theta\in\mathcal{E}_S\subset\mathbb{R}^{d_S}$ as $\theta \triangleq\left\{W_Q, W_K, W_V, W_O\right\}$, where $W_Q, W_K, W_V\in\mathbb{R}^{d_m\times d}, W_O\in\mathbb{R}^{d_m}$ are the query, key, value and output weight parameters in the Transformer. Then the output of the Transformer can be expressed as
\begin{align*}
p_\theta(\cdot~|~x)=W_O^\top \left(\frac{1}{d_S} \sum_{i=1}^{d_S} \sigma_r\left(W_V X^{\top} \sigma_s\left(\frac{X W_K^{\top} W_Q\left(X^{(i,:)}\right)^{\top}}{\sqrt{d_m}}\right)\right) \right)\in\mathbb{R}
\end{align*}
where $X^{(i,:)}$ denotes the $i$-th row of tokenization matrix $X$ for question-answer data pair $(x,y)$, $\sigma_s$ denotes the row-wise softmax and $\sigma_r$ denotes the ReLU activation function.

Following the lazy training regime in \citep{mwigogeneralization,trauger2024sequence}, we focus on the settings where the Transformer parameters remain close to their initialization throughout training, i.e. the parameters are bounded. 

\begin{theorem}[{\citep[Lemma 1]{mwigogeneralization}}]\label{thm:SFT_generalization}
Suppose $\|W_V^t\|_F\leq R_V, \|W_K^t\|_F\leq R_K,\|W_Q^t\|_F\leq R_Q,\|W_O^t\|_F\leq R_O$ holds for all iterations $t$. Also assume that the tokenized inputs for $(x_n,y_n)$ have full rank and $\|X_n\|_F \leq \sqrt{d_S} R_X$ for all $n \in[N]$, for some positive constant $R_X$. The empirical Rademacher complexity of the class of Transformer models $\mathcal{F}_\theta=\left\{x \mapsto p_\theta(\cdot ~|~ x): \left\|\theta\right\| \leq R\right\}$ can be bounded above by
\begin{align}\label{eq:C_N}
\mathcal{R}_{\mathcal{Z}_N}\left(\mathcal{F}_\theta\right) \leq \tilde{\mathcal{O}}\left(\sqrt{\frac{P_S}{N^3}}\left(1+\log \left(R_O R_V(\sqrt{d_S} R_X) \sqrt{\frac{N}{P_S}}\right)\right)\right)\triangleq C(N,R,R_X,d_S,d_{m})
\end{align}
where $\tilde{\mathcal{O}}$ hides logarithmic dependencies except $R,N$, $R=\sqrt{R_V^2+R_K^2+R_Q^2+R_O^2}$, and
\[
P_S= (\sqrt{d_S} R_X)^2\left(\left(\sqrt{d_m} R_V\right)^{\frac{2}{3}}+\left(\sqrt{d_m} R_K R_Q R_V\right)^{\frac{2}{3}}\right)^3 \log \left(N d_S\right).
\]
\end{theorem}

Theorem~\ref{thm:SFT_generalization} shows that the Rademacher complexity of the Transformer class grows with the parameter dimensions $d_m$ and $d_S$. In the low-data regime, where $N$ is small, this dimension-dependent term can dominate the bound in Theorem~\ref{thm:generalization}, leading to a looser capacity upper bound for SFT.

\subsection{Generalization bound for retriever fine-tuning}
In this section, we also specify the retrieval model as a Transformer model, which is aligned with the model we used in experiments. Similar to the SFT setting, we denote the input for a single question $x$ after tokenization as $X_R\in\mathbb{R}^{d_R\times d}$, and write the retrieval parameter $\eta\in\mathcal{E}_R\subset\mathbb{R}^{d_R}$ as $\eta \triangleq\left\{U_Q, U_K, U_V, U_O\right\}$, where $U_Q, U_K, U_V\in\mathbb{R}^{d_{m^\prime}\times d}, U_O\in\mathbb{R}^{d_{m^\prime}\times d_E}$ are the query, key, value and output weight parameters in the Transformer. The output of the Transformer can be expressed as
\begin{align*}
f_\eta (x)=U_O^\top \left(\frac{1}{d_R} \sum_{i=1}^{d_R} \sigma_r\left(U_V X_R^{\top} \sigma_s\left(\frac{X_R U_K^{\top} U_Q\left(X_R^{(i,:)}\right)^{\top}}{\sqrt{d_{m^\prime}}}\right)\right) \right)\in\mathbb{R}^{d_E}
\end{align*}
where $X_R^{(i,:)}$ denotes the $i$-th row of tokenization matrix $X_R$ for question $x$, $\sigma_s$ denotes the row-wise softmax and $\sigma_r$ denotes the ReLU activation function.

Using the learned query encoder $f_\eta(x)$, the output of RAG is
\begin{align*}
q_\eta(\cdot ~|~ x)=\sum_{d \in \mathcal{N}_k(x)} p_\eta(d \mid x) p_{\theta_0}(\cdot \mid x, d)
\end{align*}
where $\mathcal{N}_k(x)$ is the top-$k$ document set ranked by the retrieval score $s_\eta (x,d)=f_\eta(x)^\top g(d)$, and
\[
p_\eta(d \mid x)=\frac{\exp \left(s_\eta(x, d)\right)}{\sum_{d^{\prime} \in \mathcal{N}_k(x)} \exp \left(s_\eta(x, d^{\prime}\right))}.
\]

\begin{theorem}\label{thm:generalization_RAG}
Suppose $\|U_V^t\|_F\leq R_V, \|U_K^t\|_F\leq R_K,\|U_Q^t\|_F\leq R_Q,\|U_O^t\|_F\leq R_O$ holds for all iterations $t$. Also assume that the tokenized inputs for $x_n$ have full rank and $\|X_{R,n}\|_F \leq \sqrt{d_R} R_X$ for all $n \in[N]$, for some positive constant $R_X$. Additionally, assume $\|g(d)\|\leq B_g$ and $q_\eta$ is $\rho^{\text{soft}}$-Lipschitz over $s_\eta (x,d)$ on the feasible domain. The empirical Rademacher complexity of the class of retrieval-augmented predictors $\mathcal{F}_\eta=\left\{x \mapsto q_\eta(\cdot ~|~ x): \left\|\eta\right\| \leq R\right\}$ can be bounded above by
\begin{align}
\mathcal{R}_{\mathcal{Z}_N}\left(\mathcal{F}_\eta\right) \leq \tilde{\cal O}\left(\rho^{\text{soft}} B_g\sqrt{d_E}C(N,R,R_X,d_R,d_{m^\prime})\right)
\end{align}
where $\tilde{\mathcal{O}}$ hides logarithmic dependencies except $R,N$, and $C(N,R,R_X,d_R,d_{m^\prime})$ is defined in \eqref{eq:C_N}. 
\end{theorem}

\begin{proof}

First, for each coordinate $j\in [d_E]$, let us denote the
$j$-th coordinate of the query embedding as $f_{\eta,j}(x)$ and define the corresponding query-embedding class as 
\begin{align*}
\mathcal{G}_{\eta,j}=\{x \mapsto f_{\eta ,j}(x): \eta \in \mathcal{E}_R\}
\end{align*}
Then applying Theorem \ref{thm:SFT_generalization}, we know that for any $j\in [d_E]$, 
\begin{align}
\mathcal{R}_{\mathcal{Z}_N}(\mathcal{G}_{\eta,j}) \leq \tilde{\cal O}(C(N,R,R_X,d_R,d_{m^\prime}))
\end{align}
Next, define the scalar retrieval score as 
\begin{align*}
\mathcal{S}_{\eta,i}=\{x \mapsto f_{\eta}(x)^\top g(d_i): \eta \in \mathcal{E}_R\}
\end{align*}
where $i\in [k]$ and $d_i\in\mathcal{D}_k (x)$. Then, since $f_{\eta}(x)^\top g(d_i)=\sum_{j=1}^{d_E} f_{\eta,j}(x)g_j(d_i)$, and by the Cauchy-Schwarz inequality, we have 
\begin{align}
|f_{\eta}(x)^\top g(d_i)|\leq B_g \left(\sum_{j=1}^{d_E} f_{\eta,j}(x)^2\right)^{1/2}
\end{align}
Therefore, according to the definition of Rademacher complexity, we have
\begin{align}
\mathcal{R}_{\mathcal{Z}_N}(\mathcal{S}_{\eta,i})\leq B_g \left(\sum_{j=1}^{d_E} \mathcal{R}_{\mathcal{Z}_N}(\mathcal{G}_{\eta,j})^2\right)^{1/2} \leq \tilde{\cal O}(B_g\sqrt{d_E}C(N,R,R_X,d_R,d_{m^\prime}))
\end{align}
Finally, applying Talagrand's lemma for composite mappings \citep[Lemma 4.2]{mohri2018foundations} to the softmax-weighted generator mixture gives
\begin{align}
\mathcal{R}_{\mathcal{Z}_N}(\mathcal{F}_{\eta}) \leq \tilde{\cal O}(\rho^{\text{soft}}B_g\sqrt{d_E}C(N,R,R_X,d_R,d_{m^\prime}))
\end{align}
\end{proof}

\begin{remark}
Because of top-$k$ retrieval, the original $q_\eta(\cdot ~|~ x)$ might not be continuous over $\eta$, so it might not satisfy Lipschitz continuity over $s_\eta (x,d)$. 
Following \citep{basu2024statistical}, we approximate $q_\eta(\cdot ~|~ x)$ by a soft surrogate with a fixed document pool $\mathcal{D}_k(x)$ as 
\begin{align*}
q_\eta(\cdot ~|~ x)=\sum_{d \in \mathcal{D}_k (x)} p_\eta(d \mid x) p_{\theta_0}(\cdot \mid x, d)
\end{align*}
where
\[
p_\eta(d \mid x)=\frac{\exp \left(s_\eta(x, d)\right)}{\sum_{d^{\prime} \in \mathcal{D}_k(x)} \exp \left(s_\eta(x, d^{\prime}\right))}.
\]
With fixed $k$ document pools for given $x$, $q_\eta$ can be Lipschitz continuous over $s_\eta$ \citep{basu2024statistical}. 
\end{remark}

\subsection{Comparisons between generator and retriever fine-tuning}

Combining Theorems~\ref{thm:generalization} and~\ref{thm:generalization_RAG} yields Theorem~\ref{thm:gen-final}. When the number of samples is small ($N\ll d_R,d_S$), the empirical losses are comparable, and generator fine-tuning has many more trainable parameters than retriever tuning ($d_R\ll d_S$), the empirical Rademacher complexity can dominate the comparison between the two generalization upper bounds. In this regime, retriever tuning can have a smaller bound because its lower-dimensional parameterization induces a smaller complexity term.

\end{document}